\documentclass[runningheads]{llncs}
\usepackage{amssymb,amsmath}
\usepackage{graphicx,color} 
\usepackage{enumerate,paralist,xspace}
\usepackage[font=small]{subfig,caption}
\usepackage{todonotes} 
\newcommand{\pn}{\mathsf{pn}}

\let\doendproof\endproof
\renewcommand\endproof{~\hfill\qed\doendproof}

\begin{document}

\title{Low Ply Drawings of Trees\thanks{Research partially supported by DFG grant Ka812/17-1. The research by Pavel Valtr was supported by the grant GA\v{C}R 14-14179S of the Czech Science Foundation.}}

\titlerunning{Low Ply Drawings of Trees}
%
\author{
Patrizio~Angelini\inst{1} \and
Michael~A.~Bekos\inst{1} \and
Till~Bruckdorfer\inst{1} \and\\
Jaroslav Han\v{c}l Jr.\inst{2} \and
Michael~Kaufmann\inst{1}  \and
Stephen Kobourov\inst{3} \and \\
Antonios~Symvonis\inst{4} \and
Pavel~Valtr\inst{2} }
%
\authorrunning{Angelini et al.}

\institute{%
	Institut f\"ur Informatik, Universit\"at T\"ubingen, T\"ubingen, Germany
	\and
	Dept. of Applied Mathematics, Charles University (KAM), Prague, Czech Republic
	\and 
	Dept. for Computer Science, University of Arizona, Tucson, USA
	\and
	School of Applied Mathematical \& Physical Sciences, NTUA, Athens, Greece.
}

\maketitle

\begin{abstract}
We consider the recently introduced model of \emph{low ply graph drawing}, in which the ply-disks of the vertices do not have many common overlaps, which results in a good distribution of the vertices in the plane. The \emph{ply-disk} of a vertex in a straight-line drawing is the disk centered at it whose radius is half the length of its longest incident edge. The largest number of ply-disks having a common overlap is called the \emph{ply-number} of the drawing.

We focus on trees. We first consider drawings of trees with constant ply-number, proving that they may require exponential area, even for stars, and that they may not even exist for bounded-degree trees. Then, we turn our attention to drawings with logarithmic ply-number and show that trees with maximum degree $6$ always admit such drawings in polynomial area.
\end{abstract}

\section{Introduction}

Let $\Gamma$ be a straight-line drawing of a graph $G$. For a vertex $v \in G$, let the \emph{ply-disk} $D_v$ of $v$ be the open disk with center $v$ and radius $r_v$ that is half of the length of the longest incident edge of $v$. For a point $q \in \mathbb{R}^2$ in the plane, denote by $S_q$ the set of disks with $q$ in their interior, i.e., $S_q=\{D_v~|~\|v-q\| < r_v\}$. 

The \emph{ply-number} of a straight-line drawing $\Gamma$ is $\pn(\Gamma) = \max\limits_{q \in \mathbb{R}^2} |S_q|$. In other words, it describes the maximum number of ply-disks that have a common non-empty intersection. The \emph{ply-number} of a graph $G$ is $\pn(G) = \min\limits_{\Gamma \mathtt{of} G} \pn(\Gamma)$.

The ply-number is one of the most recent quality measures for graph layouts~\cite{ply-original}. While traditional measures, such as edge crossings~\cite{bcgjm-cap-13} and symmetries~\cite{eh-sgd-13}, have been studied for decades, the notion of optimizing a graph layout so that the spheres of influence of each vertex (see~\cite{liotta-proximity} for different variants) are well distributed is new. Goodrich and Eppstein~\cite{eg-snprn-08} observed that real-word geographic networks usually have only constant sphere-of-influence overlap, or in the terminology of this paper, constant ply-number. 

The problem of computing graph drawings with low ply-number is related to \emph{circle-contact representations} of graphs, where vertices are interior-disjoint circles in the plane and two vertices are adjacent if the corresponding pair of circles touch each other~\cite{h-contact-96,Hli98}. Every maximal planar graph has a circle-contact representation~\cite{Koebe36}. A drawback of such representations is that the sizes of the circles may vary exponentially, making the resulting drawings difficult to read. In \emph{balanced} circle packings and circle-contact representations, the ratio of the maximum and minimum diameters for the set of circles is polynomial in the number of vertices in the graph. Such drawings could be drawn with polynomial area, for instance, where the smallest circle determines the minimum resolution. It is known that trees and planar graphs with bounded tree-depth have balanced circle-contact representation~\cite{alam2014balanced}. Breu and Kirkpatrick~\cite{Breu19983} show that it is NP-complete to test whether a graph has a perfectly-balanced circle-contact representation, in which all circles have the same size, i.e., they are unit disks. 

Very recently, Di Giacomo et al.~\cite{ply-original} showed that binary trees, stars, and caterpillars have drawings with ply-number $2$ (with exponential area, that is, the ratio of the longest to the shortest edge is exponential in the number of vertices), while general trees with height $h$ admit drawings with ply-number $h+1$. Also, they showed that the class of graphs with ply-number $1$ coincides with the class of graphs that have a weak contact representation with unit disks, which makes the recognition problem NP-hard for general graphs~\cite{DBLP:conf/gd/FeketeHW97}. On the other hand, testing whether an internally triangulated biconnected planar graph has ply-number $1$ can be done in $O(n \log n)$ time. This paper left several natural questions open. Of particular interest are the following two questions: 

\begin{enumerate}[(i)]
\item \emph{Is it possible to draw a binary tree, a star, or a caterpillar in polynomial area with ply-number $2$?}
\item \emph{While binary trees have constant ply-number, is this true also for trees with larger bounded degree?}
\end{enumerate}
 
In this paper we provide answers to the two above questions (Section~\ref{se:area-lower-bound}). 
For the first question, we prove an exponential lower bound on the area requirements of drawings with constant ply-number of stars, and hence of caterpillars.
For the second question, we prove that there exist trees with maximum degree $11$ that do not have constant ply-number.
Motivated by these two negative results, we consider in Section~\ref{se:algorithms} drawings of trees with logarithmic ply-number. In this case, we present an algorithm to construct a drawing of every tree with maximum degree $6$ in polynomial area\footnote{The \emph{area} of a drawing is the area of the smallest axis-aligned rectangle containing it, under the resolution rule that each edge has length at least $1$.}.
We give preliminary definitions in Section~\ref{se:preliminaries} and discuss some open problems in Section~\ref{se:conclusions}.

\section{Preliminaries}\label{se:preliminaries}

Let $G$ be a graph. We denote by $\ell(e)$ (by $\ell(u,v)$) the length of an edge $e \in G$ (an edge $(u,v) \in G$) in a straight-line drawing of $G$. Also, for a path $P=v_1,\dots,v_m$, we denote by $\ell(P) = \sum_{i=1}^{m-1} \ell(v_i,v_{i+1})$ the total length of its edges.
Further, we denote by $D_v$ the ply-disk in $\Gamma$ of a vertex $v \in G$ and by $r_v$ the radius of $D_v$. Finally, we call \emph{constant-ply drawing} (or \emph{log-ply drawing}) a straight-line drawing $\Gamma$ such that $\pn(\Gamma) = O(1)$ (such that $\pn(\Gamma) = O(\log n)$).

Let $T$ be a tree rooted at a vertex $r$. The \emph{depth} $d_v$ of a vertex $v \in T$ is the length of the path between $v$ and $r$; note that $d_r=0$. The \emph{height} $h$ of $T$ is the maximum depth of a vertex of $T$.

\section{Constant-Ply Drawings of Trees}\label{se:area-lower-bound}

In this section we provide negative answers to two open questions~\cite{ply-original} about constant-ply drawings of trees. In Subsection~\ref{sse:area} we prove that drawings of this type may require exponential area, even for stars, while in Subsection~\ref{sse:10ary} we prove that there exist bounded-degree trees not admitting any of such drawings.

\subsection{Area Lower Bound for Stars}\label{sse:area}

In the original paper on the topic~\cite{ply-original}, it has been shown that a star admits a drawing with ply-number $1$ if and only if it has at most six leaves, and that every star admits a drawing with ply-number $2$, independently of the number of leaves. The algorithm for the latter result is based on a placement of the leaves at exponentially-increasing distances from the central vertex, which results in a drawing with exponential area; see Figure~\ref{fig:star-algorithm}. In this subsection we prove that this is in fact unavoidable, as we give an exponential lower bound for the area requirements of any drawing of a star with constant ply-number.

\begin{figure}[tb!]
	\centering
	\subfloat[\label{fig:star-algorithm}]{
		\includegraphics[width=0.33\textwidth,page=1]{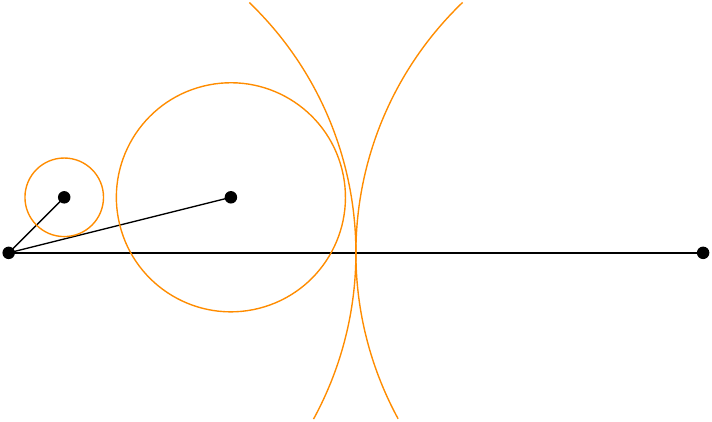}
	}\hfil
	\subfloat[\label{fig:star-lower-bound}]{
		\includegraphics[width=0.23\textwidth,page=2]{images/star.pdf}
	}
	\caption{(a) An exponential-area drawing with ply-number $2$ of a star~\cite{ply-original}. (b) Illustration for the proof of Theorem~\ref{th:star-area-lower-bound}; two disks belonging to class $\mathcal{T}_3$ are entirely contained inside annuli $\mathcal{A}_3$ and $\mathcal{A}_4$. The number close to each disk is its radius.}
	\label{fig:star}
\end{figure}

\begin{theorem}\label{th:star-area-lower-bound}
Any constant-ply drawing of an $n$-vertex star has exponential area.
\end{theorem}

\begin{proof}
Let $K_{1,n-1}$ be an $n$-vertex star with central vertex $v$, and let $\Gamma$ be a straight-line drawing of $K_{1,n-1}$ with ply-number $p$, where $p=O(1)$. We prove the statement by showing that the ratio of the longest to the shortest edge in $\Gamma$ is exponential in $n$. Assume that the longest edge $e$ of $\Gamma$ has length $\ell(e)=2$, after a possible scaling of $\Gamma$; thus, the largest ply-disk in $\Gamma$ has radius $1$.

For any $i \in \mathbb{N}$ we define $\mathcal{A}_i$ to be the annulus delimited by two circles centered at $v$ with radius $3^{-i+2}$ and $3^{-i+1}$, respectively. Refer to Figure~\ref{fig:star-lower-bound}. Then, we partition the ply-disks of the $n-1$ leaves of $K_{1,n-1}$ into the classes $\mathcal{T}_1, \dots, \mathcal{T}_k$ in such a way that all the disks with radius in $(3^{-j}, 3^{-j+1}]$ belong to $\mathcal{T}_j$, with $1 \leq j \leq k$. We observe that every disk in class $\mathcal{T}_j$ is entirely contained inside the annulus $\mathcal{A}_j \cup \mathcal{A}_{j+1}$; see Figure~\ref{fig:star-lower-bound}. However, there can be at most
$$
	\frac{p|\mathcal{A}_j \cup \mathcal{A}_{j+1}|}{\min_{D \in \mathcal{T}_j}|D|} \leq \frac{p(\pi 3^{-2j+4} - \pi 3^{-2j})}{\pi 3^{-2j}} = 80p
$$
disks in any $\mathcal{A}_j \cup \mathcal{A}_{j+1}$, and hence at most $80p$ disks belong to class $\mathcal{T}_j$. Therefore, $n = 1 + \sum_{j=1}^k |T_j| \leq 80pk$ implies that the smallest radius of the ply-disk of a vertex in a drawing is at most $3^{-k}$. This implies that the ratio between the largest and the smallest ply-disk radii in $\Gamma$, and hence between the longest and the shortest edge, is at least $3^{k} \geq 3^{n/(80p)}$. This concludes the proof.
\end{proof}


\subsection{Large Bounded-Degree Trees}\label{sse:10ary}

In this section we consider the question posed in~\cite{ply-original} on whether bounded-degree trees admit constant-ply drawings. While the answer is positive for binary trees~\cite{ply-original}, as they admit drawings with ply-number $2$, we prove that this positive result cannot be extended to all bounded-degree trees, and in particular to \emph{$10$-ary trees}, that is, rooted trees with maximum degree $11$.

\def\rot{\text{root}}

In the following we denote a complete $10$-ary tree of height $h$ by $T_{10}^h$; note that $T_{10}^h$ has $10^h$ leaves and $10^d$ vertices with depth $d \leq h$. The root of a tree $T$ is denoted by $\rot(T)$. In the rest of the section we prove the following theorem.

\begin{theorem}\label{t:10-ary trees}
For every $M>0$ there is an integer $h>0$ such that $\pn(T_{10}^h) \geq M$.
\end{theorem}

%




A \emph{branch of $T_{10}^h$} is a path in $T_{10}^h$ connecting the root with a leaf. 
Let $e$ and $f$ be two edges of $T_{10}^h$. Refer to Figure~\ref{fig:edge-dominates}. We say that $e$ \emph{dominates} $f$ and write $e>_D f$, if $e$ and $f$ lie on a common branch and $\ell(e) \geq 3^{s+1}\ell(f)$, where $s$ is the number of edges on the path between $e$ and $f$ different from $e$ and $f$. Observe that on each branch of $T_{10}^h$ the relation $>_D$ is transitive. We say that $e$ \emph{first-hand dominates} $f$ and write $e>_{\text{FD}} f$, if
the following three conditions are satisfied:
(i) $f$ lies on the path connecting $e$ with the root of $T_{10}^h$, (ii) $e$ dominates $f$, and (iii) no other edge on the path between $e$ and $f$ dominates $f$.

\begin{figure}[tb!]
	\centering
	\subfloat[\label{fig:edge-dominates}]{
		\includegraphics[height=3cm,page=1]{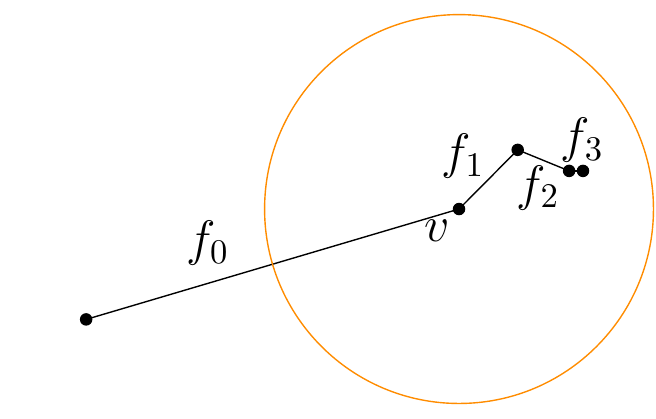}
	}\hfil
	\subfloat[\label{fig:dominate-subtree}]{
		\includegraphics[height=3cm,page=2]{images/general-trees.pdf}
	}
	\caption{(a) A path $P$ with $\ell(f_0)=28$, $\ell(f_1)=6$, $\ell(f_2)=4$, and $\ell(f_3)=1$. Edge $f_0$ dominates each of $f_1$, $f_2$, and $f_3$, which in fact lie inside $D_v$. Also, $f_0$ first-hand dominates $f_1$ and $f_2$, but does not first-hand dominate $f_3$, since $f_2$ dominates $f_3$. (b) Illustration for the proof of Lemma~\ref{l:dominatedsubtree}.}
	\label{fig:dominates}
\end{figure}

\begin{lemma}\label{le:onedominatingedge}
	Let $P$ be a path with edges $f_0,f_1,\dots,f_p$. Suppose that $f_0$ dominates each of the edges $f_1,\dots,f_p$. Let $v$ be the common vertex of the edges $f_0$ and $f_1$. Then the edges $f_1,\dots,f_p$ lie entirely inside the ply-disk $D_v$ of $v$.
\end{lemma}

\begin{proof}
See Figure~\ref{fig:edge-dominates}. Since the radius of $D_v$ is at least $\frac{\ell(f_0)}{2}$, it suffices to prove 
	$\ell(f_1)+\dots+\ell(f_p)<\frac{\ell(f_0)}{2}$. Let $i\in\{1,\dots,p\}$. Since $f_0$ dominates $f_i$, we have $\ell(f_i) \leq \frac{\ell(f_0)}{3^i}$.
	Thus,
	$\ell(f_1)+\dots+\ell(f_p) \leq \ell(f_0) (\frac{1}{3} + \frac{1}{3^2}+\cdots+\frac{1}{3^p}) < \frac{\ell(f_0)}{2}.$
\end{proof}

\begin{lemma}\label{le:Medges}
	Let $e_1,\dots,e_M$ be $M$ edges in $T_{10}^h$ such that $e_1 >_\text{FD} e_2 >_\text{FD} \cdots >_\text{FD} e_M$. Then, $\pn(T_{10}^h)\geq M.$
\end{lemma}

\begin{proof}
	By definition, $e_1,\dots,e_M$ appear in this order, possibly not consecutively, along the same branch of $T_{10}^h$. Let $\overrightarrow{P}$ be the oriented path that is the subpath of this branch from $e_1$ to $e_M$.
	Since $e_i >_\text{FD} e_{i+1}$, edge $e_i$ dominates all the edges between $e_i$ and $e_{i+1}$. Due to the transitivity of $>_\text{D}$, each edge $e_i$ dominates all the edges $e_{i+1},\dots,e_M$, and hence all the edges appearing after it along $\overrightarrow{P}$.
	
	By Lemma~\ref{le:onedominatingedge}, the endvertex $v_M$ of $e_M$ lies inside the ply-disk $D_{v_i}$ of $v_i$, for each $i=1,\dots,M$, where $v_i$ is the last vertex of $e_i$ along $\overrightarrow{P}$. Thus, the $M$ disks $D_{v_1},\dots,D_{v_M}$ have a non-empty intersection, and the statement follows.
\end{proof}

Consider a vertex $v$ with depth $d$ in $T_{10}^h$. We say that a vertex $u \neq v$ is a \emph{descendant} of $v$ if the path from $\rot(T_{10}(h))$ to $u$ contains $v$. For any $i = 1, \dots, h-d$, we denote by $T_{10}^{i}(v)$ the subtree of $T_{10}^h$ rooted at $v$ induced by $v$ and by all the descendants of $v$ with depth $d+1,d+2,\dots,d+i$. Note that $T_{10}^{i}(v)$ is a $10$-ary tree of height $i$, thus it has $10^{i}$ leaves. We have the following.

\begin{lemma}\label{l:dominatedsubtree}
	Let $T'$ be a subtree of a rooted $10$-ary tree $T$ and let $P$ be the path from $\rot(T)$ to $\rot(T')$. If every edge of $T'$ is dominated by at least one edge of $P$, then there exists a vertex $v \in P$ such that $T'$ lies completely inside $D_v$.
	
	Consequently, $\pn(T)\geq \pn(T')+1$.
\end{lemma}

\begin{proof}
	Refer to Figure~\ref{fig:dominate-subtree}. Let $e_0,e_1,\dots,e_t$ be the edges of $P$ in the order in which they appear along $P$, when $P$ is oriented from $\rot(T)$ to $\rot(T')$. Let $i$ be an index maximizing the value of $3^i\cdot\ell(e_i)$. Then $e_i$ dominates all the edges $e_{i+1},e_{i+2},\dots,e_t$. Also, due to the choice of $i$ and since any edge of $T'$ is dominated by some edge of $P$, any edge of $T'$ is dominated by $e_i$. Let $v$ be the root of $T$, if $i=0$, or the common vertex of $e_{i-1}$ and $e_{i}$ otherwise. Then Lemma~\ref{le:onedominatingedge} can be applied on the path from $v$ to any leaf of $T'$ to show that its subpath from $\rot(T')$ to the leaf lies inside $D_v$, which proves the statement.
	
	As a consequence, we have $\pn(T)\geq \pn(T')+1$. 
\end{proof}

\begin{lemma}\label{l:induction}
	Let $h,h',M$ be three positive integers such that $h'\geq h(M-1)+1$. If there exists a drawing $\Gamma$ of $T_{10}^{h'}$ that contains no $M$ edges $e_1,\dots,e_M$ such that $e_1 >_\text{FD} e_2 >_\text{FD} \cdots >_\text{FD} e_M$, then there exists a vertex $v$ in $T_{10}^{h'}$ with depth $1\leq d_v \leq h'-h$ such that no edge of $T^h_{10}(v)$ in $\Gamma$ dominates the edge $(v,v')$, where $v'$ is the neighbor of $v$ with depth $d_v-1$. Refer to Figure~\ref{fig:dominate-induction}.
\end{lemma}

\begin{proof}
	We fix $h$ and proceed by induction on $M$. If $M=1$, then there exists no drawing $\Gamma$ of $T_{10}^{h'}$ satisfying the conditions of the lemma, and thus the statement holds. Suppose now that $M>1$ and that the lemma holds for $M-1$. We want to show that the lemma holds for $M$. Let $h'\geq h(M-1)+1$. Suppose that a drawing of $T_{10}^{h'}$ contains no $M$ edges $e_1,\dots,e_M$ such that $e_1 >_\text{FD} e_2 >_\text{FD} \cdots >_\text{FD} e_M$. 
	
	\begin{figure}[tb!]
		\centering
		\subfloat[\label{fig:dominate-induction}]{
			\includegraphics[height=3cm,page=3]{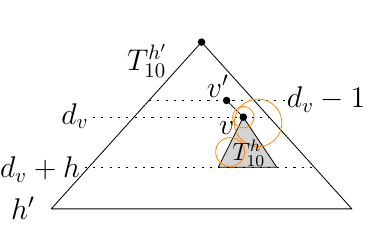}
		}\hfil
		\subfloat[\label{fig:dominate-case-2}]{
			\includegraphics[height=3cm,page=4]{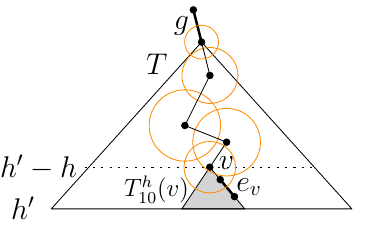}
		}
		\caption{(a) Illustration for Lemma~\ref{l:induction}; no edge of $T^h_{10}(v)$ dominates edge $(v,v')$. (b) Illustration for {\bf Case 2}: for each vertex $v$ with depth $h'-h$, there exists an edge $e_v$ in $T^h_{10}(v)$ that is not dominated by any edge of path from $v$ to the root of $T$.}
		\label{fig:induction}
	\end{figure}
	
	Consider the subtree $T':=T_{10}^{h'-h}(\rot(T_{10}^{h'}))$ of $T_{10}^{h'}$, with the same root as $T_{10}^{h'}$, that is induced by the vertices with depth at most $h'-h$.
	If $T'$ does not contain $M-1$ edges $e_1,\dots,e_{M-1}$ such that $e_1 >_\text{FD} e_2 >_\text{FD} \cdots >_\text{FD} e_{M-1}$, then the required vertex $v$ exists by induction.
	Otherwise, consider $M-1$ edges $e_2,\dots,e_{M}$ in $T'$ such that $e_2 >_\text{FD} e_3 >_\text{FD} \cdots >_\text{FD} e_{M}$. Let $d$ and $d+1$, with $d < h'-h$, be the depth of the endvertices $v'$ and $v$ of $e_2$, respectively, in $T'$ (it is the same depth as they have in $T_{10}^{h'}$). Consider the subtree $T^h_{10}(v)$ of $T_{10}^{h'}$ rooted at $v$. Suppose, for a contradiction, that there exists an edge in $T^h_{10}(v)$ dominating $e_2$. Then, consider the edge $e_1$ in $T^h_{10}(v)$ dominating $e_2$ with the property that no other edge on the path from $e_1$ to $e_2$ dominates $e_2$, that is, $e_1$ first-hand dominates $e_2$. Thus, $e_1 >_\text{FD} e_2 >_\text{FD} \cdots >_\text{FD} e_M$, a contradiction. This implies that no edge of $T^h_{10}(v)$ dominates edge $e_2$, and the statement follows.
\end{proof}

We are now ready to complete the proof of the main result of the section.

\begin{proof}[of Theorem~\ref{t:10-ary trees}]
	We proceed by induction on $M$. For $M=1$ the statement trivially holds.
	
	Suppose now that $M>1$ and that for $M-1$ there is an $h$ with the required properties.
	We need to show that for $M$ there is an $h'$ with the required properties.
	We set $h':=\max\{h^2M,Ch(h+M)\}$, where $C$ is a (large) constant to be specified later. We fix a drawing of $T_{10}^{h'}$.
	
	If there are $M$ edges $e_1,\dots,e_M$ in $T_{10}^{h'}$ such that $e_1 >_\text{FD} e_2 >_\text{FD} \cdots >_\text{FD} e_M$, then Lemma~\ref{le:Medges} implies $\pn(T_{10}^{h'})\geq M$. Otherwise, due to Lemma~\ref{l:induction} there is a rooted $10$-ary subtree $T$ of $T_{10}^{h'}$ with height $\overline{h} \geq \frac{h'}{M}$ such that $\rot(T)\neq \rot(T_{10}^{h'})$ and no edge of $T$ dominates the first edge on the path from $\rot(T)$ to $\rot(T_{10}^{h'})$. From now on, we focus on the rooted tree $T$. In particular, in the following we refer to the depth of a vertex as its depth in $T$.
	We distinguish two cases.
	
	In {\bf Case 1} there exists a vertex $v$ with depth $\overline{h}-h$ in $T$ such that every edge of the tree $T_{10}^{h}(v)$ is dominated by at least one edge of the path from $v$ to $\rot(T)$. In this case, Lemma~\ref{l:dominatedsubtree} (applied on tree $T_{10}^{h}(v)$) and the inductive hypothesis show that $\pn(T_{10}^{h'})\geq \pn(T_{10}^{h}(v))+1\geq M$.
	
	In {\bf Case 2} there exists no vertex in $T$ with the above properties. Refer to Figure~\ref{fig:dominate-case-2}. Thus, for any vertex $v$ with depth $\overline{h}-h$ in $T$, the subtree $T_{10}^{h}(v)$ rooted at $v$ contains at least one edge that is not dominated by any edge of the path from $v$ to $\rot(T)$; among these edges of $T_{10}^{h}(v)$ we choose one, denoted by $e(v)$, whose endvertices have the smallest possible depth. This implies that $e(v)$ is not dominated by any edge of the path $P_v$ from its endvertex $u_v$ to $\rot(T)$.
	
	Let $g$ be the first edge from $\rot(T)$ to $\rot(T_{10}^{h'})$. Note that edges of $P_v$ dominate neither $g$ nor $e(v)$. W.l.o.g., assume $\ell(g)=1$. Since edges $g$ and $e(v)$ do not dominate each other, we have
	$ 1/3^{\overline{h}} < \ell(e(v)) < 3^{\overline{h}}. $
	Thus, there is a unique integer $k(v)\in\{-\overline{h},-\overline{h}+1,\dots,\overline{h}-1\}$ such that $\ell(e(v))\in [3^{k(v)},3^{k(v)+1})$.
	
	Let $k$ be a most frequent value of $k(v)$ over all the vertices $v$ with depth $\overline{h}-h$. Since $k(v)$ may have $2\overline{h}$ different values, the set $V_k$ of vertices $v$ at level $\overline{h}-h$ with $k(v)=k$ has size at least $10^{\overline{h}-h}/(2\overline{h})$. Consider now a vertex $v\in V_k$ and the path $P_v$ from $\rot(T)$ to $u_v$. Since no edge of this path dominates
	$g$ or $e(v)$, we have the following two upper bounds on the length of the $i$-th edge $e_i$ of the path $P_v$ oriented from $\rot(T)$ to $u_v$:	
	$$\ell(e_i) \leq 3^i\ \ \ \text{and}$$
	$$ \ell(e_i) \leq 3^{\overline{h}-i} \cdot \ell(e(v)) < 3^{\overline{h}-i+k+1}.$$
	
	For the latter, we use $\ell(e(v)) < 3^{k+1}$, which follows from the fact that $v\in V_k$. The edges $e_i$ with $i\le (\overline{h}+k)/2$ have total length at most
	$\sum_{i=0}^{\lfloor (\overline{h}+k)/2 \rfloor}3^i\leq 3^{(\overline{h}+k)/2+1}$,
	and the total length of the other edges is at most
	$$\sum_{i=\lfloor (\overline{h}+k)/2+1\rfloor}^{\overline{h}} 3^{\overline{h}-i+k+1}
	=  3^{k+1} \cdot  \sum_{i=\lfloor (\overline{h}+k)/2\rfloor+1}^{\overline{h}} 3^{\overline{h}-i} $$
	$$ =  3^{k+1} \cdot  \sum_{j=0}^{\overline{h}-\lfloor (\overline{h}+k)/2\rfloor-1} 3^j
	\leq 3^{k+1} \cdot 3^{(\overline{h}-k)/2+1}
	= 3^{(\overline{h}+k)/2+2}. $$
	It follows that the total length of the path $P_v$ is smaller than
	$12\cdot 3^{(\overline{h}+k)/2}.$
	
	Thus all the edges $e(v),v\in V_k$, lie in the disk $D$ of radius $12\cdot 3^{(\overline{h}+k)/2}$ centered at $\rot(T)$.
	The area of $D$ is $12^2\pi3^{\overline{h}+k}$. Let $v\in V_k$, and let $u'_v$ be the vertex of the path $P_v$ adjacent to $u_v$. The ply-disk $D_{u'_v}$ contains the disk of radius $3^k/2$ centered at $u'_v$, which is entirely contained in $D$. It follows that the region $D_{u'_v}\cap D$ has area at least $\pi(3^k/2)^2=(\pi/4)3^{2k}$.
	Therefore there is a point of $D$ lying in at least
	$$ \frac{|V_k|\cdot(\pi/4)3^{2k}}{\text{area}(D)}
	\geq \frac{(10^{\overline{h}-h}/(2\overline{h}))\cdot(\pi/4)3^{2k}}{12^2\pi3^{\overline{h}+k}}
	=   \frac{(10/3)^{\overline{h}}/\overline{h}\cdot3^k}{12^2\cdot 8\cdot 10^h}
	\geq \frac{(10/9)^{\overline{h}}/\overline{h}}{12^2\cdot 8\cdot 10^h}$$
	disks $D_{u'_v}$, with $v\in V_k$.
	
	Since $h'\geq CM(h+\log M)$,  we have $\overline{h}\ge C(h+\log M)$. If $C$ is a sufficiently large constant
	then some point of $D$ lies in at least
	$$ \frac{({10/9})^{\overline{h}}/\overline{h}}{12^2\cdot 8\cdot 10^h}\geq M$$
	disks $D_{u'_v}$, with $v\in V_k$, which concludes the proof.
\end{proof}

\section{Log-Ply Drawings of Bounded-Degree Trees in Polynomial Area}\label{se:algorithms}

Motivated by the fact that constant-ply drawings of stars may require exponential area (Theorem~\ref{th:star-area-lower-bound}) and by the fact that not all the bounded-degree trees admit a constant-ply drawing (Theorem~\ref{t:10-ary trees}), in this section we ask whether allowing a logarithmic ply-number makes it possible to always construct drawings of trees, possibly in polynomial area. 
We give a first answer by proving in Theorem~\ref{thm:treelogn} that this is true for $5$-ary trees, that is, trees with maximum degree $6$. 

We start with some definitions. A \emph{$2$-drawing} of a path $P=v_1,\dots,v_m$ is a straight-line drawing of $P$ in which all the vertices lie along the same straight-line segment in the same order as they appear in $P$ and for each $i=2,\dots,m$ we have $\frac{\ell(v_{i-1},v_{i})}{2} \leq \ell(v_i,v_{i+1}) \leq 2\ell(v_{i-1},v_{i})$; see Figure~\ref{fig:2-path}. We have the following.

\begin{lemma}\label{le:2-drawing-ply-2}
A $2$-drawing of a path $P= (v_1,\dots,v_n)$ has ply-number at most $2$.
\end{lemma}
\begin{proof}
Refer to Figure~\ref{fig:2-path}. For each vertex $v_i$, we have radius $r_{v_i} \leq \ell(v_i,v_{i+1})$ and $r_{v_i} \leq \ell(v_{i-1},v_i)$, since $\frac{\ell(v_{i-1},v_{i})}{2} \leq \ell(v_i,v_{i+1}) \leq 2\ell(v_{i-1},v_{i})$. This, together with the fact that all the vertices of $P$ lie along the same straight-line segment, implies that the ply-disk $D_{v_i}$ of $v_i$ may only intersect with $D_{v_{i-1}}$ and with $D_{v_{i+1}}$, but not with any of the other disks (note that $D_{v_i}$ may touch $D_{v_{i-2}}$ and $D_{v_{i+2}}$ in a single point, namely the one where vertices $v_{i-1}$ and $v_{i+1}$ lie, respectively), but cannot overlap with them.
\end{proof}

\begin{figure}[tb!]
	\centering
	\subfloat[\label{fig:2-path}]{
		\includegraphics[width=0.4\textwidth,page=1]{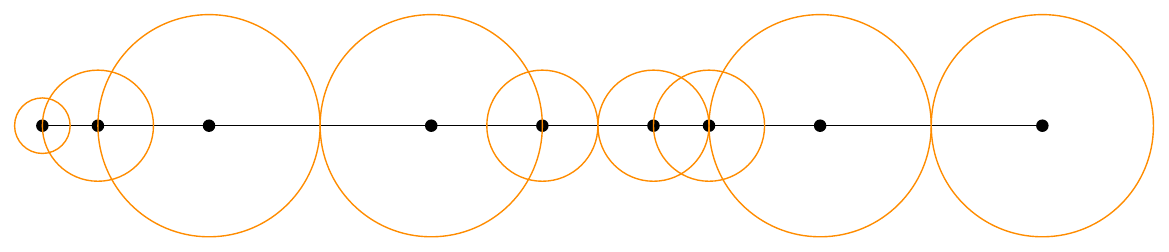}
	}\hfil
	\subfloat[\label{fig:heavy}]{
		\includegraphics[width=0.26\textwidth,page=1]{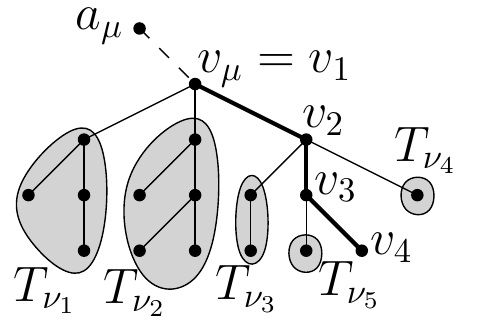}
	}\hfil
	\subfloat[\label{fig:heavy-tree}]{
		\includegraphics[width=0.26\textwidth,page=2]{images/heavy-path-tree.pdf}
	}
	\caption{(a) A $2$-drawing of a path. (b) A ternary tree $T_\mu$ and the path $\mu$, represented by fat edges, that is a node of the heavy-path tree $\mathcal{T}$; the subtrees $T_{\nu_1},\dots,T_{\nu_5}$ obtained when removing $\mu$ are inside shaded region. (c) The portion of $\mathcal{T}$ containing nodes $\mu$ and its children $\nu_1,\dots,\nu_5$. The arc of $\mathcal{T}$ between $\mu$ and a node $\nu_i$ is labeled with the node of $\mu$ that is the anchor of $\nu_i$.}
	\label{fig:heavy-path-tree}
\end{figure}

The \emph{heavy-path tree} $\mathcal{T}$ of a rooted tree $T$ is a decomposition tree of $T$ first defined by Sleator and Tarjan~\cite{st-adsdt-83} as follows; see Figures~\ref{fig:heavy}-\ref{fig:heavy-tree}. Each node $\mu \in \mathcal{T}$ is a path in $T$ between a vertex $v_\mu$ of $T$ and a leaf of the subtree $T_\mu$ of $T$ rooted at $v_\mu$. At the first step, $v_\mu$ is the root of $T$, $T_\mu$ is $T$, and the path $\mu$ we construct is the root of $\mathcal{T}$. To construct $\mu$, we start from $v_\mu$ and we always select the child of the current vertex whose subtree contains the largest number of vertices, until a leaf of $T_\mu$ is reached. Then, we remove all the vertices of $\mu$ from $T_\mu$ and their incident edges, hence obtaining a set of subtrees of $T_\mu$. For each of these subtrees $T_\nu$, rooted at a vertex $v_\nu$, we add a new node $\nu$ as a child of $\mu$ in $\mathcal{T}$ and recursively construct the corresponding path. Since each subtree $T_\nu$ has at most half of the vertices of $T_\mu$, the height of the heavy-path tree $\mathcal{T}$ is $O(\log n)$.

Let $\mu = (v_\mu=v_1, \dots,v_m)$ be any node in $\mathcal{T}$ and let $\tau$ be its parent. The vertex of $\tau$ that is adjacent to $v_\mu$ is the \emph{anchor} $a_\mu$ of $\mu$; in order to have an anchor $a_\mu$ also when $\mu$ is the root of $\mathcal{T}$, we add a dummy vertex to $T$ that is only incident to its root. The proof of the main theorem of this section is based on the following algorithm, which we call {\sc DrawPath}, to construct a special $2$-drawing of the path $P$ that is the concatenation of edge $(a_\mu,v_\mu)$ and of path~$\mu$.

Let $n_\mu$ be the total number of vertices in the subtrees of $T_\tau$ whose corresponding paths have $a_\mu$ as an anchor. Since $T_\mu$ is one of these subtrees, we have that $n_\mu > \sum_{i=1}^m n_i$, where $n_i$ is the total number of vertices in the subtrees $T_{\nu_1},\dots,T_{\nu_h}$ of $T_\mu$ such that paths $\nu_1,\dots,\nu_h$ have $v_i$ as anchor. Also, since $\mu$ is a path in a heavy-path tree, we have $n_i \leq n_\mu/2$ for each $1 \leq i \leq m$.

Algorithm {\sc DrawPath} starts by initializing $\ell(a_\mu,v_1)=n_1$ and $\ell(v_i,v_{i+1})=n_i+n_{i+1}$, for each $i=1,\dots,m-1$. Then, it visits the edges of $P$ one by one in decreasing order of their length in the current drawing. When an edge $(v_i,v_{i+1})$, with $1 \leq i \leq m-1$, is visited, set $\ell(v_{i-1},v_i)=\max\{\frac{\ell(v_i,v_{i+1})}{2},\ell(v_{i-1},v_i)\}$ and $\ell(v_{i+1},v_{i+2})=\max\{\frac{\ell(v_i,v_{i+1})}{2},\ell(v_{i+1},v_{i+2})\}$. We have the following.

\begin{lemma}\label{le:path-draw-algo-ternary}
Algorithm {\sc DrawPath} constructs
a $2$-drawing $\Gamma$ of $P$ such that $\ell(a_\mu,v_1) \geq n_1$, $\ell(v_i,v_{i+1}) \geq n_i+n_{i+1}$, for each $i=1,\dots,m-1$, and $\ell(P) \leq 6 n_\mu$.
\end{lemma}
\begin{proof}
First observe that $\ell(a_\mu,v_1) \geq n_1$ and $\ell(v_i,v_{i+1}) \geq n_i+n_{i+1}$ for each $i=1,\dots,m-1$, since this is true already after the initialization and since no operation performed by the algorithm reduces the length of any edge.

Also, the fact that $\Gamma$ is a $2$-drawing can be derived from the operations that are performed when an edge is visited. Note that after an edge has been visited by {\sc DrawPath}, its length is not modified any longer, since the edges are visited in decreasing order of edge lengths and since the length of an edge is modified only if this edge is shorter than one of its adjacent edges.
	
For the same reason, if an edge $(v_i,v_{i+1})$, with $1 \leq i \leq m-1$, determines a local maximum in the sequence of edge lengths in $\Gamma$ (that is, $\ell(v_h,v_{h+1}) \geq \ell(v_{h-1},v_h)$ and $\ell(v_h,v_{h+1}) \geq \ell(v_{h+1},v_{h+2})$), then $\ell(v_i,v_{i+1})=n_i+n_{i+1}$. We use this property to prove that the total length of the edges in $\Gamma$ is at most $6 n_\mu$. 
	
Consider any two edges $(v_h,v_{h+1})$ and $(v_q,v_{q+1})$, with $1 \leq h < q \leq m-1$, such that $\ell(v_h,v_{h+1})=n_h+n_{h+1}$, $\ell(v_q,v_{q+1})=n_q+n_{q+1}$, and such that $\ell(v_i,v_{i+1}) > n_i+n_{i+1}$ for each $i=h+1,\dots,q-1$; namely, $(v_h,v_{h+1})$ and $(v_q,v_{q+1})$ are two edges that have not been modified by algorithm {\sc DrawPath} after the initialization and such that all edges between them have been modified.

\begin{claim}
The total length of the edges in the subpath $P'$ of $\Gamma$ between $v_h$ and $v_{q+1}$ is at most $2(n_h+n_{h+1}) + 2(n_q+n_{q+1})$.
\end{claim}
\begin{proof}
Note that there exists no edge in $P'$ different from $(v_h,v_{h+1})$ and $(v_q,v_{q+1})$ that determines a local maximum in the sequence of edge lengths, since this would contradict the fact that $\ell(v_i,v_{i+1})>n_i+n_{i+1}$ for each $i=h+1,\dots,q-1$. Hence, $P'$ is composed of a sequence of edges starting at $(v_h,v_{h+1})$ and ending at an edge $(v_{j-1},v_{j})$, with $h < j \leq q-1$, with decreasing edge lengths, and of a sequence of edges starting at $(v_j,v_{j+1})$ and ending at $(v_q,v_{q+1})$ with increasing edge lengths. 
We have $\ell(v_h,v_{h+1})=n_h+n_{h+1}$ and $\ell(v_q,v_{q+1})=n_q+n_{q+1}$, by construction. Also, $\sum_{i=h+1}^{j-1} \ell(v_i,v_{i+1}) = \sum_{i=1}^{j-1-h} \frac{n_h+n_{h+1}}{2^{i}} < n_h+n_{h+1}$, since $\Gamma$ is a $2$-drawing.
Analogously, $\sum_{i=j}^{q-1} \ell(v_i,v_{i+1}) < n_q+n_{q+1}$. 
\end{proof}
	
Hence, every edge $(v_h,v_{h+1})$ such that $\ell(v_h,v_{h+1})=n_h+n_{h+1}$, together with the possible sequence of edges with increasing (decreasing) edge lengths preceding (following) it, gives a contribution of less than $3(n_h+n_{h+1})$. Since $\sum_{i=1}^{m} (n_i+n_{i+1}) < 2 n_\mu$, the total edge length is at most $6 n_\mu$.
%
\end{proof}

We describe an algorithm to construct a log-ply drawing of any rooted $n$-vertex $5$-ary tree $T$ with polynomial area. To simplify the description, we first give the algorithm for ternary trees; we discuss later the extension to $5$-ary trees.

Construct the heavy-path tree $\mathcal{T}$ of $T$. Then, construct a drawing of $T$ recursively according to a bottom-up traversal of $\mathcal{T}$. At each step of the traversal, consider a path $\mu \in \mathcal{T}$. We associate $\mu$ with a half-disk $D_\mu$ of radius $6^{h-d_\mu} n_\mu$, where $h$ is the height of $\mathcal{T}$ and $d_\mu$ is the depth of $\mu$. Refer to Figure~\ref{fig:half-disk}. The goal is to construct a drawing with ply-number at most $2 (h-d_\mu+1)$ of the subtree $T_\mu$ rooted at $v_\mu$, augmented with the anchor $a_\mu$ of $\mu$ and with edge $(a_\mu,v_\mu)$, inside $D_\mu$ in such a way that $a_\mu$ lies on the center of $D_\mu$ and all the vertices of $\mu$ lie along the radius of $D_\mu$ that is perpendicular to the diameter delimiting $D_\mu$.

\begin{figure}[tb!]
	\centering
	\subfloat[\label{fig:half-disk}]{
		\includegraphics[height=3cm,page=2]{images/trees-algorithm-ternary.pdf}
	}\hfil
	\subfloat[\label{fig:algo-ternary}]{
		\includegraphics[height=3.4cm,page=1]{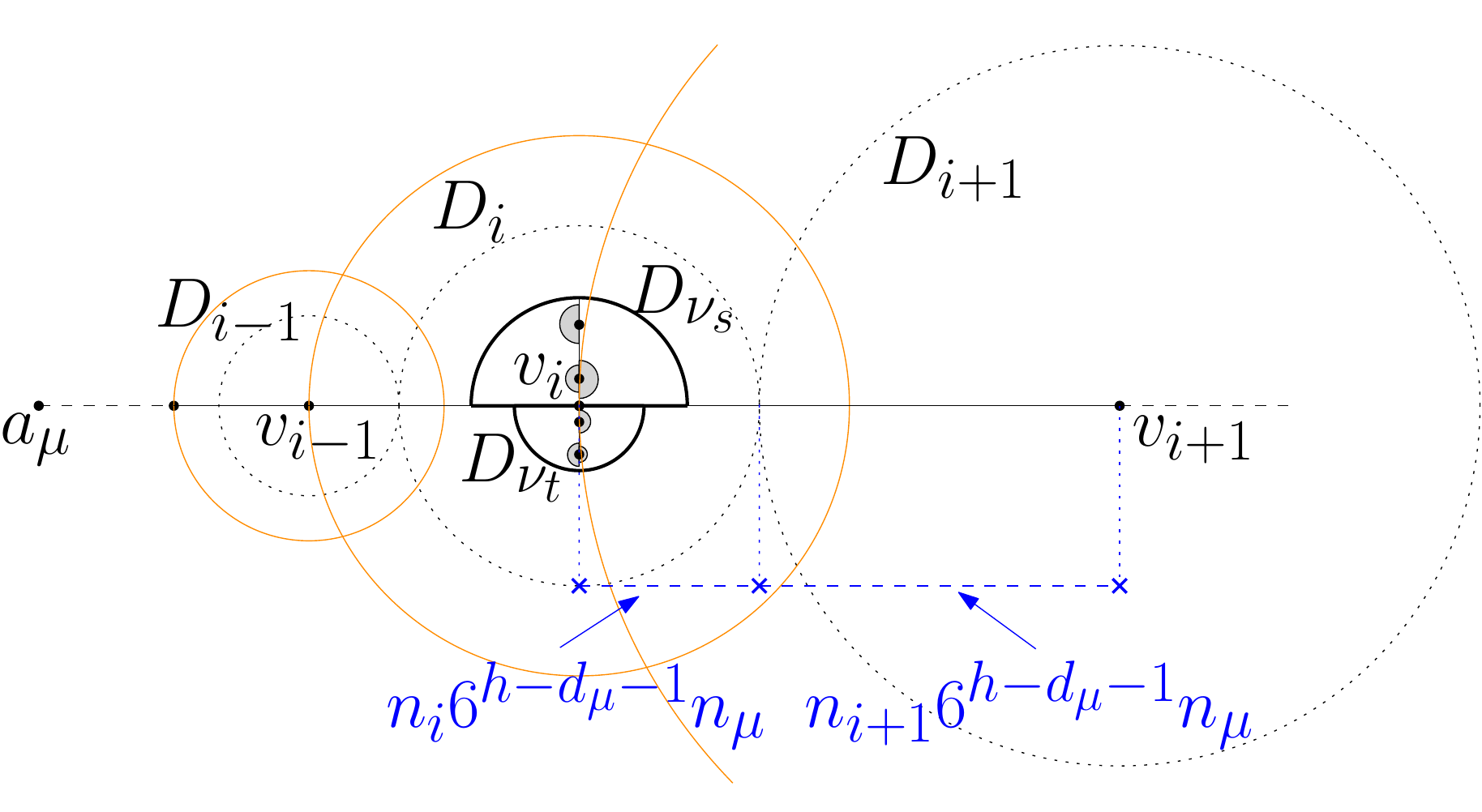}
	}\hfil
	\subfloat[\label{fig:quarter}]{
		\includegraphics[height=3cm,page=3]{images/trees-algorithm-ternary.pdf}
	}
	\caption{(a) The half-disk $D_\mu$ associated with $\mu$. (b) Illustration for the algorithm for ternary trees. Black dotted circles are the disks $D_i$; orange solid circles are the ply-disks. (c) Using a quarter-disk instead of a half-disk for $5$-ary trees.}
	\label{fig:algorithm}
\end{figure}

If $\mu=(v_1,\dots,v_m)$ is a leaf of $\mathcal{T}$, place $a_\mu$ on the center of $D_\mu$ and the vertices of $\mu$ along the radius of $D_\mu$ perpendicular to the diameter delimiting it, so that each edge has length $1$. This drawing has ply-number $1$ and satisfies the required properties by construction.

If $\mu=(v_1,\dots,v_m)$ is not a leaf, let $\nu_1,\dots,\nu_k$ be its children. Assume inductively that for each child $\nu_j$, with $j=1,\dots,k$, there exists a drawing with ply-number at most $2 (h-d_{\nu_j}+1)$ inside the half-disk $D_{\nu_j}$ with radius $6^{h-d_{\nu_j}} n_{\nu_j}$ with the required properties. We show how to construct a drawing with ply-number at most $2 (h-d_{\mu}+1)$ of $T_\mu$ inside the half-disk $D_{\mu}$ with radius $6^{h-d_{\mu}} n_{\mu}$ with the required properties; recall that $d_{\mu} = d_{\nu_j}-1$, for each $j=1,\dots,k$.

Refer to Figure~\ref{fig:algo-ternary}. First, apply algorithm {\sc DrawPath} to construct a $2$-drawing of the path $P$ composed of $\mu$ and of its anchor $a_\mu$ such that $\ell(a_\mu,v_1) \geq n_1$, $\ell(v_i,v_{i+1}) \geq n_i+n_{i+1}$, for $i=1,\dots,m-1$, and the total length of the edges in $P$ is at most $6 n_\mu$. Then, scale the obtained drawing by a factor of $6^{h-d_{\mu}-1} n_{\mu}$, which implies that the total length of the edges in $P$ is at most $6^{h-d_{\mu}} n_{\mu}$. Hence, it is possible to place the obtained drawing inside $D_\mu$ in such a way that $a_\mu$ lies on its center and the vertices of $\mu$ lie along the radius that is perpendicular to the diameter delimiting $D_\mu$. Further, for each vertex $v_i \in \mu$, consider a disk $D_i$ centered at $v_i$ of diameter $6^{h-d_{\mu}-1} n_{i}$. Due to the scaling performed before, no two disks $D_i$ and $D_h$, with $1 \leq i,h \leq m$, intersect with each other. 

Consider now the at most two children $\nu_s$ and $\nu_t$ of $\mu$ whose anchor is $v_i$; since $d_{\mu} = d_{\nu_j}-1$, for each $j=1,\dots,k$, and since $n_i = n_{\nu_s}+n_{\nu_t}$, the diameter of the half-disk $D_{\nu_s}$ and the one of the half-disk $D_{\nu_t}$ are both not larger than the diameter of disk $D_i$. Thus, we can plug the drawings of $T_{\nu_s}$ and $T_{\nu_t}$ lying inside $D_{\nu_s}$ and  $D_{\nu_t}$, which exist by induction, so that the centers of $D_{\nu_s}$ and  $D_{\nu_t}$ coincide with the center of $D_i$, and the diameters delimiting $D_{\nu_s}$ and  $D_{\nu_t}$ lie along edges $(v_{i-1},v_i)$ and $(v_i,v_{i+1})$; see Figure~\ref{fig:algo-ternary}. Hence, the constructed drawing of $T_\mu$ lies inside $D_\mu$ and satisfies all the required properties.

By Lemma~\ref{le:2-drawing-ply-2}, the ply-number of the $2$-drawing of $\mu$ constructed by algorithm {\sc DrawPath} is at most $2$, and it remains the same after the scaling. Also, the ply-disk of any vertex in $T_{\nu_s}$ (in $T_{\nu_t}$) entirely lies inside half-disk $D_{\nu_s}$ (half-disk $D_{\nu_t}$) and hence inside disk $D_i$; thus, it does not overlap with the ply-disk of any vertex in a different subtree. Since the drawing of $T_{\nu_j}$, for each child $\nu_j$ of $\mu$, has ply-number at most $2 (h-d_{\nu_j}+1)$, the drawing of $T_\mu$ has ply-number at most $2 + 2 (h-d_{\nu_j}+1) = 2 (h-d_{\nu_j}+2) = 2 (h-d_{\mu}+1)$, given that $d_{\mu} = d_{\nu_j}+1$.

At the end of the traversal, when the root $\rho$ of $\mathcal{T}$ has been visited, we have a drawing with ply-number at most $2(h-d_{\rho}+1) \leq 2\log n$ of $T_\rho=T$ inside the half-disk $D_\rho$ of radius $6^{h-d_{\rho}} n_{\rho} \leq 6^{\log n} n = O(n^{1+\log 6})=O(n^{3.6})$, and hence area $O(n^{7.2})$.

In order to extend the algorithm to work for $5$-ary trees, we have to be able to fit inside the ply-disk $D_i$ of each vertex $v_i \in \mu$ the drawings of the at most four subtrees $T_{\nu_j}$ whose anchor is $v_i$. Hence, we associate with each node $\mu$ a \emph{quarter-disk} $D_\mu$ (a sector of a disk with internal angle $\frac{\pi}{2}$; see Figure~\ref{fig:quarter}), instead of a half-disk, still with radius $6^{h-d_\mu} n_\mu$, and we draw $T_\mu$ inside $D_\mu$ in such a way that the anchor $a_\mu$ of $\mu$ lies on the center of $D_\mu$ and all the vertices of $\mu$ lie along the radius of $D_\mu$ along the bisector of $D_\mu$. Also in this case, the ply-disk of each vertex of $\mu$ entirely lies inside $D_\mu$. We thus have the following.

\begin{theorem}\label{thm:treelogn}
Every $n$-vertex $5$-ary tree has a drawing with ply-number at most $2 \log n$ and $O(n^{7.2})$ area.
\end{theorem}

To extend this approach for trees with larger degree, we should use a disk sector $D_\mu$ with an internal angle smaller than $\frac{\pi}{2}$. In this case, however, we could not guarantee that the ply-disk of each vertex of $\mu$ lies inside $D_\mu$, and thus we could not compute the ply-number of the subtrees independently of each other.

\section{Conclusions and Open Problems}\label{se:conclusions}

In this work we considered drawings of trees with low ply-number. We proved that requiring the ply-number to be bounded by a constant is often a somewhat too strong limitation, as these drawings may not exist, even for bounded-degree trees, or may require exponential area. On the positive side, we showed that relaxing the requirement on the ply-number, allowing it to be bounded by a logarithmic function, makes the problem easier, as we gave an algorithm for constructing polynomial-area drawings with this property for trees with maximum degree $6$.

Our work leaves several interesting open questions:

First, while it is known that stars, caterpillars, and binary trees admit constant-ply drawings in exponential area~\cite{ply-original}, we were able to prove that this is unavoidable only for stars and caterpillars; this leaves open the question on the area-requirements of constant-ply drawings of binary trees.

Second, it would be interesting to reduce the gap between binary trees, which alway admit constant-ply drawings, and $10$-ary trees, which may not admit any of such drawings. More in general, a characterization of the trees admitting these drawings is a fundamental open question.

Finally, in this paper we provided the first results on log-ply drawings of trees. It would be worth studying which trees (or other classes of graphs) always admit this type of drawings, possibly with polynomial area.

\clearpage

\bibliographystyle{splncs03}
\bibliography{biblio}

\end{document}